\documentclass[runningheads]{llncs}
\usepackage[T1]{fontenc}
\usepackage[misc,geometry]{ifsym}
\usepackage{comment,amsmath,amssymb,xspace,xcolor,tikz,graphicx,tabularx,tablefootnote,boxedminipage}
\usetikzlibrary{fit,automata,arrows}
\usetikzlibrary{decorations.pathreplacing}
\usetikzlibrary{decorations.pathmorphing}
\usetikzlibrary{shapes,snakes}

\newtheorem{open}{Open Problem}
\newcommand{\optproblemdef}[3]{
	\begin{center}
		\begin{boxedminipage}{.99\textwidth}
			\textsc{{#1}}\\[2pt]
			\begin{tabular}{ r p{0.8\textwidth}}
				\textit{~~~~Instance:} & {#2}\\
				\textit{Task:} & {#3}
			\end{tabular}
		\end{boxedminipage}
	\end{center}
}
\newcommand{\problemdef}[3]{
	\begin{center}
		\begin{boxedminipage}{.99\textwidth}
			\textsc{{#1}}\\[2pt]
			\begin{tabular}{ r p{0.8\textwidth}}
				\textit{~~~~Instance:} & {#2}\\
				\textit{Question:} & {#3}
			\end{tabular}
		\end{boxedminipage}
	\end{center}
}
\newcommand{\NP}{{\sf NP}}
\newcommand{\ssi}{\subseteq_i}

\title{Classifying Subset Feedback Vertex Set\\ for $H$-Free Graphs\thanks{A preliminary version of this paper will appear in the proceedings of WG 2022~\cite{PPR}.}}

\author{Giacomo Paesani\inst{1}\inst{1}\orcidID{0000-0002-2383-1339} 
\and
Dani\"el Paulusma\inst{2}\orcidID{0000-0001-5945-9287}
\and
Pawe{\l} Rz{\k a}\.{z}ewski\inst{3,4}\orcidID{0000-0001-7696-3848}\thanks{Supported by Polish National Science Centre grant no. 2018/31/D/ST6/00062.}}

\authorrunning{G. Paesani, D. Paulusma, P. Rz{\k a}\.{z}ewski}

\institute{
School of Computing, University of Leeds, UK, \\ \email{g.paesani@leeds.ac.uk}
\and Department of Computer Science, Durham University, UK, \\ \email{daniel.paulusma@durham.ac.uk}
\and Faculty of Mathematics and Information Science, Warsaw University of Technology, Poland \and Faculty of Mathematics, Informatics, and Mechanics, University of Warsaw, Poland
\email{pawel.rzazewski@pw.edu.pl}}

\begin{document}

\maketitle

\begin{abstract} 
In the {\sc Feedback Vertex Set} problem, we aim to find a small set $S$ of vertices in a graph intersecting every cycle.
The {\sc Subset Feedback Vertex Set} problem requires $S$ to intersect only those cycles that include a vertex of some specified set $T$.
We also consider the {\sc Weighted Subset Feedback Vertex Set} problem, where each vertex~$u$ has weight $w(u)>0$
and we ask that~$S$ has small weight.
By combining known \NP-hardness results with new polynomial-time results we prove full complexity dichotomies for {\sc Subset Feedback Vertex Set} and {\sc Weighted Subset Feedback Vertex Set} for $H$-free graphs, that is, graphs that do not contain a graph~$H$ as an induced subgraph.  

\end{abstract}

\section{Introduction}\label{s-intro}

In a {\it graph transversal} problem the aim is to find a small set of vertices within a given graph that must intersect every subgraph that belongs to some specified family of graphs. 
Apart from the {\sc Vertex Cover} problem, the {\sc Feedback Vertex Set} problem is perhaps the best-known graph transversal problem. A vertex subset $S$ is a {\it feedback vertex set} of a graph~$G$ if $S$ intersects every cycle of $G$. In other words, the graph $G-S$ obtained by deleting all vertices of $S$ is a forest. We can now define the problem:
 
 \problemdef{{\sc Feedback Vertex Set}}
{a graph $G$ and an integer $k$.}
{does $G$ have a feedback vertex set $S$ with $|S|\leq k$?} 

\noindent
The {\sc Feedback Vertex Set} problem is well-known to be \NP-complete even under input restrictions. For example, by Poljak's construction~\cite{Po74}, the {\sc Feedback Vertex Set} problem is \NP-complete even for graphs of finite girth at least~$g$ (the girth of a graph is the length of its shortest cycle). To give another relevant example, {\sc Feedback Vertex Set} is also \NP-complete for line graphs~\cite{Mu17b}. 

In order to understand the computational hardness of {\sc Feedback Vertex Set} better, other graph classes have been considered as well, in particular those that are closed under vertex deletion. Such graph classes are called {\it hereditary}. It is readily seen that a graph class ${\cal G}$ is hereditary if and only if ${\cal G}$ can be characterized by a (possibly infinite) set ${\cal F}$ of forbidden induced subgraphs. From a systematic point of view it is natural to first consider the case where ${\cal F}$ has size~$1$, say ${\cal F}=\{H\}$ for some graph $H$. This leads to the notion of $H$-freeness: a graph $G$ is {\it $H$-free} for some graph $H$ if $G$ does not contain $H$ as an {\it induced} subgraph, that is, $G$ cannot be modified into $H$ by a sequence of vertex deletions.

As {\sc Feedback Vertex Set} is \NP-complete for graphs of finite girth at least~$g$ for every $g\geq 1$, it is \NP-complete for $H$-free graphs whenever $H$ has a cycle. As it is \NP-complete for line graphs and line graphs are claw-free, {\sc Feedback Vertex Set} is \NP-complete for $H$-free graphs whenever $H$ has an induced claw (the {\it claw} is the $4$-vertex star). 
In the remaining cases, the graph~$H$ is a {\it linear forest}, that is, the disjoint union of one or more paths. When $H$ is a linear forest, several positive results are known even for the weighted case. Namely, for a graph $G$, we can define a \emph{(positive) weighting} as a function 
$w: V\to \mathbb{Q}^+$.
For $v \in V$, $w(v)$ is the \emph{weight} of $v$, and for $S \subseteq V$, we define the weight $w(S)=\sum_{u\in S}w(u)$ of $S$ as the sum of the weights of the vertices in~$S$. This brings us to the following generalization of {\sc Feedback Vertex Set}:
 
\problemdef{{\sc Weighted Feedback Vertex Set}}
{a graph $G$, a positive vertex weighting~$w$ of $G$ and a rational number $k$.}
{does $G$ have a feedback vertex set $S$ with $w(S)\leq k$?} 

\noindent
Note that if $w$ is a constant weighting function,
then we obtain the {\sc Feedback Vertex Set} problem.
We denote the $r$-vertex path by $P_r$, and the {\it disjoint union} of two vertex-disjoint graphs $G_1$ and $G_2$ by $G_1+G_2=((V(G_1)\cup V(G_2),E(G_1)\cup E(G_2))$, where we write $sG$ for the disjoint union of $s$ copies of $G$.
It is known that {\sc Weighted Feedback Vertex Set} is polynomial-time solvable for $sP_3$-free graphs~\cite{PPR21} and $P_5$-free graphs~\cite{ACPRS20}. The latter result was recently extended to
$(sP_1+P_5)$-free graphs for every $s\geq 0$~\cite{PPR21}.
We write $H\ssi G$ to denote that $H$ is an {\it induced} subgraph of $G$. We can now summarize all known results~\cite{ACPRS20,Mu17b,PPR21,Po74} as follows.

\begin{theorem}\label{t-kn}
{\sc (Weighted) Feedback Vertex Set} for the class of $H$-free graphs is polynomial-time solvable if $H\ssi sP_3$ or $H\ssi sP_1+P_5$ for some $s\geq 1$, and is \NP-complete if $H$ is not a linear forest.
\end{theorem}

\noindent
Note that the open cases of Theorem~\ref{t-kn} are when $H$ is a linear forest with $P_2+P_4\ssi H$ or $P_6\ssi H$.

The {\sc (Weighted) Feedback Vertex Set} problem can be further generalized in the following way. Let $T$ be some specified subset of vertices of a graph $G$. A {\it $T$-cycle} of $G$ is a cycle that intersects~$T$.  A set $S_T\subseteq V$~is a {\it $T$-feedback vertex set} of $G$ if $S_T$ contains at least one vertex of every $T$-cycle; see also Fig.~\ref{f-example}.
We now consider the following generalizations of {\sc Feedback Vertex Set}:

\problemdef{{\sc Subset Feedback Vertex Set}}
{a graph $G$, a subset $T\subseteq V(G)$ and an integer $k$.}
{does $G$ have a $T$-feedback vertex set $S_T$ with $|S_T|\leq k$?} 

\problemdef{{\sc Weighted Subset Feedback Vertex Set}}
{a graph $G$, a subset $T\subseteq V(G)$, a positive vertex weighting~$w$ of $G$ and a rational number $k$.}
{does $G$ have a $T$-feedback vertex set $S_T$ with $w(S_T)\leq k$?} 

\begin{figure}[b]
\begin{center}
\begin{minipage}{0.45\textwidth}
\centering
\begin{tikzpicture}[xscale=0.55, yscale=0.55]
\draw (0,2)--(1.16,-1.6)--(-1.9,0.6)--(1.9,0.6)--(-1.16,-1.6)--(0,2) (-2.85,0.9)--(-1.74,-2.4)--(1.74,-2.4)--(2.85,0.9)--(0,3)--(-2.85,0.9) (-1.16,-1.6)--(1.16,-1.6)
(-1.9,0.6)--(-2.85,0.9) (-1.16,-1.6)--(-1.74,-2.4) (1.16,-1.6)--(1.74,-2.4) (1.9,0.6)--(2.85,0.9) (0,2)--(0,3);
\draw[fill=white](-1.9,0.6) circle [radius=5pt] (-1.16,-1.6) circle [radius=5pt] (1.16,-1.6) circle [radius=5pt] (1.9,0.6) circle [radius=5pt] 
(0,2) node[regular polygon,regular polygon sides=4,draw,fill=black,scale=0.7pt] {} (-2.85,0.9) node[regular polygon,regular polygon sides=4,draw,fill=white,scale=0.7pt] {} 
(-1.74,-2.4) node[regular polygon,regular polygon sides=4,draw,fill=black,scale=0.7pt] {} (1.74,-2.4) node[regular polygon,regular polygon sides=4,draw,fill=white,scale=0.7pt] {} (0,3) circle [radius=5pt];
\draw[fill=black](2.85,0.9) circle [radius=5pt];
\end{tikzpicture}
\end{minipage}
\qquad
\begin{minipage}{0.45\textwidth}
\centering
\begin{tikzpicture}[xscale=0.55, yscale=0.55]
\draw (0,2)--(1.16,-1.6)--(-1.9,0.6)--(1.9,0.6)--(-1.16,-1.6)--(0,2) (-2.85,0.9)--(-1.74,-2.4)--(1.74,-2.4)--(2.85,0.9)--(0,3)--(-2.85,0.9) (-1.16,-1.6)--(1.16,-1.6)
(-1.9,0.6)--(-2.85,0.9) (-1.16,-1.6)--(-1.74,-2.4) (1.16,-1.6)--(1.74,-2.4) (1.9,0.6)--(2.85,0.9) (0,2)--(0,3);
\draw[fill=white] (-1.9,0.6) circle [radius=5pt] (-1.16,-1.6) circle [radius=5pt] (1.16,-1.6) circle [radius=5pt] (1.9,0.6) circle [radius=5pt]
(2.85,0.9) circle [radius=5pt] (0,2) node[regular polygon,regular polygon sides=4,draw,fill=black,scale=0.7pt] {}
(-2.85,0.9) node[regular polygon,regular polygon sides=4,draw,fill=black,scale=0.7pt] {} (-1.74,-2.4) node[regular polygon,regular polygon sides=4,draw,fill=white,scale=0.7pt] {}
(1.74,-2.4) node[regular polygon,regular polygon sides=4,draw,fill=black,scale=0.7pt] {}(0,3) circle [radius=5pt];
\end{tikzpicture}
\end{minipage}
\caption{Two examples of a slightly modified Petersen graph with the set $T$ indicated by square vertices. In both examples, the set $S_T$ of black vertices is a $T$-feedback vertex set. On the left, $S_T\setminus T\neq \emptyset$.
On the right, $S_T\subseteq T$.}\label{f-example}.
\end{center}
\end{figure}
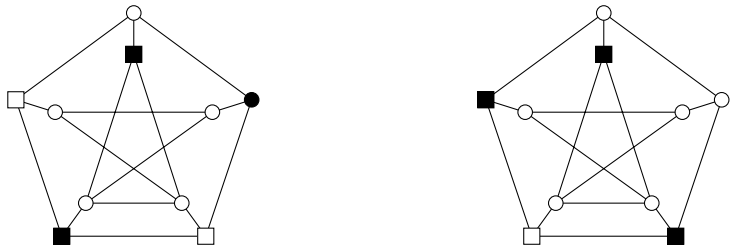

\noindent
The \NP-complete cases in Theorem~\ref{t-kn} carry over to {\sc (Weighted) Subset Feedback Vertex Set}; just set $T:=V(G)$ in both cases. However, this is no longer true for the polynomial-time cases: Fomin et al.~\cite{FHKPV14} proved \NP-completeness of {\sc Subset Feedback Vertex Set} for split graphs, which form a subclass of $2P_2$-free graphs. Interestingly,
Papadopoulos and Tzimas~\cite{PT20} proved that {\sc Weighted Subset Feedback Vertex Set} is \NP-complete for $5P_1$-free graphs, whereas Brettell et al.~\cite{BJPP20}
proved that {\sc Subset Feedback Vertex Set} can be solved in polynomial time even for $(sP_1+P_3)$-free graphs for every $s\geq 1$~\cite{BJPP20}. Hence,
in contrast to many other transversal problems, the complexities on the weighted and unweighted subset versions do not coincide for $H$-free graphs. 
 
It is also known that {\sc Weighted Subset Feedback Vertex Set} can be solved in polynomial time for permutation graphs~\cite{PT19} and thus for its subclass of $P_4$-free graphs. The latter result also follows from a more general result related to the graph parameter mim-width~\cite{Va12}. Namely, Bergougnoux, Papadopoulos and Telle~\cite{BPT19} proved that {\sc Weighted Subset Feedback Vertex Set} is polynomial-time solvable for graphs for which we can find a decomposition of constant mim-width in polynomial time~\cite{BPT19}; the class of $P_4$-free graphs is an example of such a class. 
Brettell et al.~\cite{BJP21} extended these results by proving that {\sc Weighted Subset Feedback Vertex Set}, restricted to $H$-free graphs, is polynomial-time solvable if 
$H\ssi 3P_1+P_2$ or $H\ssi P_1+P_3$. 

The above results leave open a number of unresolved cases for both problems, as identified in~\cite{BJPP20} and~\cite{BJP21}, where the following open problems are posed:

\begin{figure}
\begin{center}
\begin{tikzpicture}[scale=0.8] \draw (1.5,0)--(-1.5,0); \draw[fill=black] (1.5,0) circle [radius=3pt] (0.5,0) circle [radius=3pt] (-0.5,0) circle [radius=3pt] (-1.5,0) circle [radius=3pt] (0.8,1) circle [radius=3pt] (-0.8,1) circle [radius=3pt];
\end{tikzpicture}
\end{center}
\caption{The graph $2P_1+P_4$.}\label{f-2p1p4}
\end{figure}
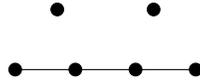

\begin{open}\label{o-1}
Determine the complexity of {\sc Weighted Subset Feedback Vertex Set} for $H$-free graphs if $H\in \{2P_1+P_3,P_1+P_4,2P_1+P_4\}$.
\end{open}

\begin{open}\label{o-2}
Determine the complexity of {\sc Subset Feedback Vertex Set} for $H$-free graphs if $H=sP_1+P_4$ for some integer $s\geq 1$.
\end{open}

\subsection{Our Results}\label{s-ours}

We completely solve Open Problems~\ref{o-1} and~\ref{o-2}.

In Section~\ref{s-2p1p4}, we prove that {\sc Weighted Subset Feedback Vertex Set} is polynomial-time solvable for $(2P_1+P_4)$-free graphs. This result generalizes all known polynomial-time results for {\sc Weighted Feedback Vertex Set}. It also immediately implies polynomial-time solvability for the other two cases in Open Problem~\ref{o-1}, as $(2P_1+P_3)$-free graphs and $(P_1+P_4)$-free graphs form subclasses of $(2P_1+P_4)$-free graphs. Combining the aforementioned \NP-completeness results of~\cite{FHKPV14} and~\cite{PT20} for $2P_2$-free graphs and $5P_1$-free graphs, respectively, with the \NP-completeness results in Theorem~\ref{t-kn} for the case where $H$ has a cycle or a claw and this new result gives us the following complexity dichotomy (see also Fig.~\ref{f-2p1p4}).

\begin{theorem}\label{t-dicho}
For a graph~$H$, {\sc Weighted Subset Feedback Vertex Set} on $H$-free graphs is polynomial-time solvable if 
$H\ssi 2P_1+P_4$, and is \NP-complete otherwise.
\end{theorem}

\noindent
In Section~\ref{s-2p1p4u} we solve Open Problem~\ref{o-2} by proving that {\sc Subset Feedback Vertex Set} is polynomial-time solvable for $(sP_1+P_4)$-free graphs, for every $s\geq 1$.
This result generalizes all known polynomial-time results for {\sc Weighted Feedback Vertex Set}.
After combining it with the aforementioned \NP-completeness results of~\cite{FHKPV14} and Theorem~\ref{t-kn} we obtain the following complexity dichotomy.

\begin{theorem}\label{t-dicho2}
For a graph~$H$, {\sc Subset Feedback Vertex Set} on $H$-free graphs is polynomial-time solvable if 
$H\ssi sP_1+P_4$ for some $s\geq 0$, and is \NP-complete otherwise.
\end{theorem}

\noindent
Due to Theorems~\ref{t-dicho} and~\ref{t-dicho2} we now know where exactly the complexity jump between the weighted and unweighted versions occurs. 

Our proof technique for these results is based on the following two ideas. First, if  the complement $F_T$ of a $T$-feedback vertex set contains $s$ vertices of small degree in $F_T$, then we can ``guess'' these vertices and their neighbours in $F_T$. We then show that after removing all the other neighbours of small-degree vertices, we will obtain a graph of small mim-width. 
If $F_T$ does not contain $s$ small-degree vertices, we will argue that $F_T$ contains a bounded number of vertices of $T$. We guess these vertices and exploit their presence. This is straightforward for {\sc Subset Feedback Vertex Set} but more involved for {\sc Weighted Subset Feedback Vertex Set}. The latter was to be expected from the hardness construction for {\sc Weighted Subset Feedback Vertex Set} on $5P_1$-free graphs, in which $|T|=1$ (but as we will show our algorithm is able to deal with that construction due to the $(2P_1+P_4)$-freeness of the input graph).

We finish our paper with a brief discussion on related graph transversal problems and some open questions in Section~\ref{s-con}.

\section{Preliminaries}\label{s-pre}

Let $G=(V,E)$ be a graph. If $S\subseteq V$, then~$G[S]$ denotes the subgraph~of~$G$ induced by~$S$, and~$G-S$ is the graph $G[V\setminus S]$. 
We say that~$S$ is {\it independent} if $G[S]$ has no edges, and that $S$ is a {\it clique} and $G[S]$ is {\it complete} if every pair of vertices in $S$ is joined by an edge. A {\it  (connected) component} of~$G$ is a maximal connected subgraph of $G$. The \emph{neighbourhood} of a vertex $u\in V$ is the set $N(u)=\{v\; |\; uv\in E\}$. A graph is \emph{bipartite} if its vertex set can be partitioned into at most two independent sets.

Recall that for a graph $G=(V,E)$ and a subset $T \subseteq V$, a $T$-feedback vertex set is a set $S \subseteq V$ that intersects all $T$-cycles. Note that $G-S$ is a graph that has no $T$-cycles; we call such a graph a {\it $T$-forest}.
Thus the problem of finding a $T$-feedback vertex set of minimum size is equivalent to finding a $T$-forest of maximum size.
Similarly, the problem of finding a $T$-feedback vertex set of minimum weight is equivalent to finding a $T$-forest of maximum weight.
These maximisation problems are actually the problems that we will solve. Consequently, any $T$-forest will be called a {\it solution} 
for an instance $(G,T)$ or $(G,w,T)$, respectively, and our aim is to find a solution of maximum size or maximum weight, respectively.

Throughout our proofs we will need to check if some graph $F$ is a solution. The following lemma  shows that we can recognize solutions in linear time. The lemma combines results claimed but not proved in~\cite{LMRS17,PT20}. It is easy to show but for an explicit proof we refer to \cite[Lemma~3]{BJPP20}.

\begin{lemma}\label{l-test}
It is possible to decide in $O(n+m)$ time if a graph $F$ is a $T$-forest for some given set $T\subseteq V(F)$.
\end{lemma}

\noindent
In our proofs we will not refer to Lemma~\ref{l-test} explicitly, but we will use it implicitly every time we must check if some graph $F$ is a solution.

\section{The Weighted Variant}\label{s-2p1p4}

In this section, we present our polynomial-time algorithm for {\sc Weighted Subset Feedback Vertex Set} on $(2P_1+P_4)$-free graphs. 

\smallskip
\noindent
{\bf Outline.}
Our algorithm is based on the following steps. We first show in Section~\ref{s-special} how to compute a solution~$F$ that contains at most one vertex from~$T$, which moreover has small degree in $F$.  In Section~\ref{s-mim} we then show that if two vertices of small degree in a solution are non-adjacent, we can exploit the $(2P_1+P_4)$-freeness of the input graph $G$ to reduce to a graph $G'$ of bounded mim-width. The latter enables us to apply the algorithm of  Bergougnoux, Papadopoulos and Telle~\cite{BPT19}. In Section~\ref{s-algo} we deal with the remaining case, where all the vertices of small degree in a solution $F$ form a clique and $F$ contains at least two vertices of $T$. We first show that every vertex of $T$ that belongs to $F$ must have small degree in $F$. Hence, as the vertices in $T\cap V(F)$ must also induce a forest, $F$ has exactly two adjacent vertices of $T$, each of small degree in $F$. This structural result enables us to do a small case analysis.
We combine this step together with our previous algorithmic procedures into one algorithm.

\medskip
\noindent
{\bf Remark.}
Some of the lemmas in the following three subsections hold for $(sP_1+P_4)$-free graphs, for every $s\geq 2$, or even for general graphs. In order to re-use these lemmas in Section~\ref{s-2p1p4u}, where we consider {\sc Subset Feedback Vertex Set} for $(sP_1+P_4)$-free graphs, we formulate these lemmas as general as possible.

\subsection{Three Special Types of Solutions}\label{s-special}

In this section we will show how we can find three special types of solutions in polynomial time for $(2P_1+P_4)$-free graphs. These solutions have in common that they contain at most one vertex from the set~$T$ and moreover, this vertex has small degree in $F$.

Let $G=(V,E)$ be a graph and let $T\subseteq V$ be a subset of vertices of $G$. A $T$-forest~$F$ is a {\it $\leq$$1$-part solution} if $F$ contains at most one vertex from $T$ and moreover, if $F$ contains a vertex~$u$ from $T$, then $u$ has degree at most~$1$ in $F$. The following lemma holds for general graphs and is easy to see, but we need it for our algorithm.

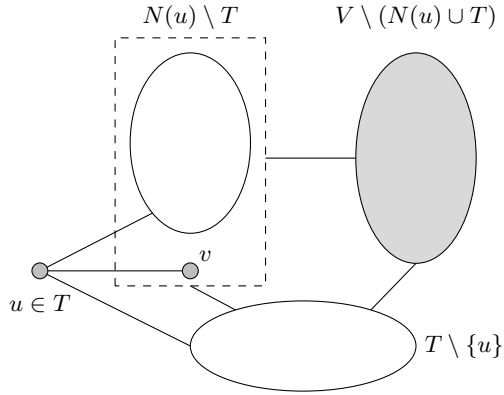
\begin{figure}
\begin{center}
\begin{tikzpicture}[scale=1] \draw[color=black,fill=white] (-3,3.7) ellipse (0.8cm and 1.2cm); \draw[color=black,fill=white] (-1.5,1) ellipse (1.5cm and 0.6cm); \draw[color=black,fill=gray!30!white] (0,3.5) ellipse (0.8cm and 1.4cm); \draw[dashed] (-4,1.8)--(-4,5.1)--(-2,5.1)--(-2,1.8)--(-4,1.8); \draw (-3.5,2.77)--(-5,2)--(-3,2) (-5,2)--(-3,1) (-3,1.8)--(-2.4,1.48) (-2,3.5)--(-0.8,3.5) (0,2.1)--(-0.6,1.48); \draw[fill=gray!50!white] (-5,2) circle [radius=3pt]; \draw[fill=gray!50!white] (-3,2) circle [radius=3pt]; \node[above] at (-3,5.1) {$N(u)\setminus T$}; \node[right] at (0,1) {$T\setminus \{u\}$}; \node[above] at (0,5.1) {$V\setminus (N(u)\cup T)$}; \node[right] at (-3,2.2) {$v$}; \node[below] at (-5,1.8) {$u\in T$}; \end{tikzpicture}
\end{center}
\caption{The $\leq$$1$-part solution (in grey) constructed in the proof of Lemma~\ref{l-1part}. Such a solution consists of the vertex $u\in T$, the set $V\setminus (N(u)\cup T)$ and possibly a vertex $v\in N(u)\setminus T$.}\label{f-1}
\end{figure}

\begin{lemma}\label{l-1part}
For a graph $G=(V,E)$ with a positive vertex weighting~$w$ and a set $T\subseteq V$, it is possible to find a $\leq$$1$-part solution of maximum weight in polynomial time.
\end{lemma} 

\begin{proof}
The solution $F=G-T$ has maximum weight over all solutions that contain no vertex of $T$.
Let $|V|=n$. We now consider all $O(n)$ possibilities of choosing a vertex $u\in T$ to be in the $\leq$$1$-part solution $F$ that we are looking for. Consider $u\in T$ to be in $F$. In the case where $u$ has no neighbour in $F$, we can set $F$ to be the subgraph of $G$ induced by $u$ and all non-neighbours of $u$ in $V\setminus T$.
Otherwise, if $N(u)\setminus T\neq \emptyset$, then we consider all $O(n)$ possibilities of choosing a neighbour $v\in N(u)\setminus T$ 
to be in $F$ as well. Afterwards, we define $F$ to be the subgraph induced by $u$, $v$ and all non-neighbours of $u$ in $V\setminus T$; see also Fig.~\ref{f-1}. It is readily seen that this algorithm is correct and runs in polynomial time.
\qed
\end{proof}

Let $G=(V,E)$ be a graph and let $T\subseteq V$ be a subset of vertices of $G$. A $T$-forest~$F$ is a {\it $2$-part solution} if $F$ contains exactly one vertex $u$ of $T$ and $u$ has exactly two neighbours $v_1$ and $v_2$ in $F$. We say that $u$ is the {\it center} of $F$ and that $v_1$ and $v_2$ are the {\it center neighbours}. Let $A$ be the connected component of $F$ that contains $u$. Then we say that $A$ is the {\it center component} of~$F$.
 We make the following observation.

\begin{lemma}\label{l-sep}
Let $A$ be the center component of a $2$-part solution $F$ with center~$u$. Then $A-u$ consists of exactly two connected components (each of which containing a different center neighbour). 
\end{lemma}

\begin{proof}
Note that $A-u$ can contain at most two connected components as $u$ has degree~$2$ in $A$.
For a contradiction, assume that $A-u$ is connected. Then $A-u$ contains a path between the two center neighbours of $u$. This implies that $u\in T$ is on a cycle in $A$, contradicting the fact that $F$ is a $T$-forest.
\qed
\end{proof}

We will prove how to find $2$-part solutions in polynomial time even for general graphs. In order to do this, we will reduce to a classical problem, namely:

\optproblemdef{{\sc Weighted Vertex Cut}}{a graph $G=(V,E)$, two distinct 
non-adjacent terminals $t_1$ and $t_2$, and a positive vertex weighting~$w$.}{determine a set $S\subseteq V\setminus \{t_1,t_2\}$ of minimum weight such that  $t_1$ and $t_2$ are in different connected components of $G-S$.}

\noindent
The {\sc Weighted Vertex Cut} problem is well known to be polynomial-time solvable by standard network flow techniques.

\begin{figure}
\begin{center}
\begin{tikzpicture}[scale=1] \draw[color=black,fill=white] (-3,3.7) ellipse (0.8cm and 1.2cm); \draw[color=black,fill=white] (1,4.5) ellipse (1.5cm and 0.6cm); \draw[color=black,fill=gray!30!white] (1,2.5) ellipse (0.8cm and 0.8cm); \draw[color=black,fill=gray!30!white] (1,1.5) ellipse (0.8cm and 0.8cm); \draw[color=black,fill=white] (1,2) ellipse (0.7cm and 0.5cm); \draw[dotted] (-3,2)--(-3,1); \draw[dashed] (-4,0.8)--(-4,5.1)--(-2,5.1)--(-2,0.8)--(-4,0.8); \draw (-3.5,2.77)--(-5,2)--(-3,2) (-5,2)--(-3,1) (0.2,2.5)--(-3,2) (0.2,1.5)--(-3,1) (-0.5,4.5)--(-2,4) (1,3.9)--(1,3.3); \draw[fill=gray!50!white] (-5,2) circle [radius=3pt] (-3,2) circle [radius=3pt] (-3,1) circle [radius=3pt]; \node[above] at (-3,5.1) {$N(u)\setminus T$}; \node[above] at (1,5.1) {$T\setminus \{u\}$}; \node[right] at (2.2,2) {$V\setminus (N(u)\cup T)$}; \node[right] at (-3,2.2) {$v_1$}; \node[right] at (-3,1.2) {$v_2$}; \node[below] at (-5.1,1.75) {$u\in T$};
\end{tikzpicture}
\end{center}
\caption{The $2$-part solution (in grey) constructed in the proof of Lemma~\ref{l-2part}. The graph $G'$ consists of the vertex $u$, the independent set $\{v_1,v_2\}$ and the set $V\setminus (N(u)\cup T)$. The $2$-part solution is obtained from $G'$ after deleting the vertices of a minimum weight vertex cut of $G'-u$ that separates $v_1$ from $v_2$.}\label{f-2}
\end{figure}
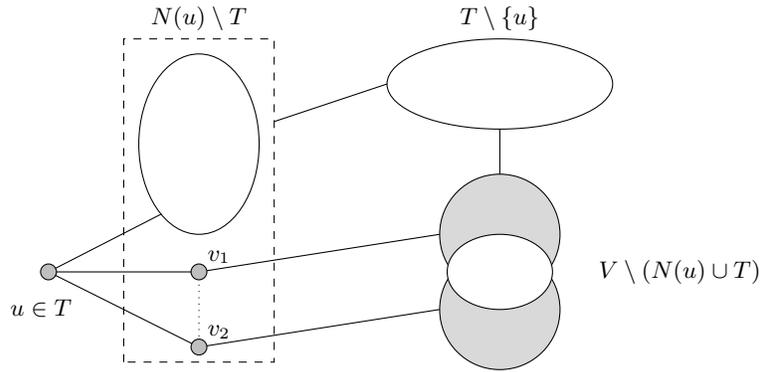

\begin{lemma}\label{l-easy}
{\sc Weighted Vertex Cut} is polynomial-time solvable.
\end{lemma}

\noindent
We use Lemma~\ref{l-easy} to prove the following lemma.

\begin{lemma}\label{l-2part}
For a graph $G=(V,E)$ with a positive vertex weighting~$w$ and a set $T\subseteq V$, it is possible to find a $2$-part solution of maximum weight in polynomial time.
\end{lemma}

\begin{proof}
Let $|V|=n$. We consider every possible choice of center~$u$ and center neigbours $v_1,v_2$. Note there are $O(n^3)$ options. For each choice of $u,v_1,v_2$ we do as follows.
First we check that $u\in T$, $\{v_1,v_2\}\cap T=\emptyset$, $u$ is adjacent to both $v_1$ and $v_2$ and there is no edge between $v_1$ and $v_2$. If the check fails, the guess is discarded.
We remove every neighbour of $u$ that is not in $\{v_1,v_2\}$. We also remove every vertex of $T\setminus \{u\}$ from $G$, as $u$ will be the only vertex of $F$ that belongs to $T$.
We denote the resulting graph by $G'$. By Lemma~\ref{l-sep}, it remains to solve {\sc Weighted Vertex Cut} for the instance that consists of $G'-u$, with $v_1$ and $v_2$ as terminals and the restriction of $w$ to $V(G')\setminus \{u\}$; see also Fig.~\ref{f-2}.
By Lemma~\ref{l-easy}, this can be done in polynomial time. In the end we return a solution of maximum weight. As the total number of different options to consider is $O(n^3)$, our algorithm runs in polynomial time.
\qed
\end{proof}

\noindent
Let $G=(V,E)$ be a graph and let $T\subseteq V$ be a subset of vertices of $G$.  A $T$-forest~$F$ is a {\it $3$-part solution} if $F$ contains exactly one vertex $u$ of $T$ and $u$ has exactly three neighbours $v_1,v_2,v_3$ in $F$. Again we say that $u$ is the {\it center} of $F$; that $v_1,v_2,v_3$ are the {\it center neighbours}; and that the connected component of~$F$ that contains $u$ is the {\it center component} of~$F$.
We make the same observation as we did in Lemma~\ref{l-sep}.

\begin{lemma}\label{l-sep2}
Let $A$ be the center component of a $3$-part solution $F$ with center~$u$. Then $A-u$ consists of exactly three connected components (each of which containing a different center neighbour). 
\end{lemma}

\begin{proof}
Note that $A-u$ can contain at most three connected components as $u$ has degree~$3$ in $A$. For a contradiction, assume that $A-u$ contains at most two connected components. Then one of these components contains at least two of the three center neighbours of $u$.
Then $A-u$ contains a path between these two center neighbours. This implies that $u\in T$ is on a cycle in $A$, contradicting the fact that $F$ is a $T$-forest.
\qed
\end{proof}

Let $A$ be the center component of a $3$-part solution $F$ with center~$u$. By Lemma~\ref{l-sep2}, we have that $A-u$ consists of exactly three connected components $D_1$, $D_2$ and $D_3$. We say that $F$ is {\it full} if every $D_i$ $(i=1,2,3)$ has size at least~$2$. We now prove the following structural lemma for $(2P_1+P_4)$-free graphs.

\begin{lemma}\label{l-full}
Let $G$ be a $(2P_1+P_4)$-free graph.
Let $A$ be the center component of a full $3$-part solution $F$ for $G$ that has center~$u$. Then $F=A$ and every connected component of $A-u$ is a complete graph.
\end{lemma}

\begin{proof}
Let $v_1,v_2,v_3$ be the center neighbours of $F$.
By Lemma~\ref{l-sep2}, we have that $A-u$ consists of exactly three connected components $D_1$, $D_2$ and $D_3$, where we assume that $v_i\in V(D_i)$ for $i\in \{1,2,3\}$. As $F$ is full, we have that for $i\in \{1,2,3\}$, $D_i$ contains a vertex $x_i$ adjacent to $v_i$.

We will first prove that $F=A$. Suppose that $F$ contains a vertex $y\notin V(A)$. Then $x_1,v_1,u,v_2,x_3,y$ induce a $2P_1+P_4$ in $G$, contradicting the $(2P_1+P_4)$-freeness of $G$.

We will now prove that every $D_i$ is complete. We consider $D_1$, as the arguments for $D_2$ and $D_3$ will be the same. First suppose that $D_1$ contains a vertex $z$ not adjacent to $v_1$. 
As $D_1$ is connected we may assume without loss of generality that $z$ has distance~$2$ from $v_1$ in $D_1$, and so $D_1$ contains a vertex $z'$ that is adjacent to both $v_1$ and $z$. 
Then $z,z',v_1,u,x_2,x_3$ induce a $2P_1+P_4$ in $G$, a contradiction. Hence, $v_1$ is adjacent to every other vertex in $D_1$. Now suppose that $D_1$ still contains two non-adjacent vertices $p$ and $q$. As $v_1$ is adjacent to every vertex of $V(D_1)\setminus \{v_1\}$, we find that $v_1\notin \{p,q\}$. Then $x_2,v_2,u,v_3,p,q$ induce a $2P_1+P_4$ in $G$, a contradiction.
\qed
\end{proof}

\noindent
We are now ready to prove the following lemma.

\begin{lemma}\label{l-3part}
For a $(2P_1+P_4)$-free graph $G=(V,E)$ with a positive vertex weighting~$w$ and a set $T\subseteq V$, it is possible to find a $3$-part solution of maximum weight in polynomial time.
\end{lemma} 

\begin{proof}
Let $|V|=n$.
We consider every possible choice of center~$u$ and center neighbours $v_1,v_2,v_3$. Note that there are $O(n^4)$ options. For each choice of $u,v_1,v_2,v_3$ we do as follows. 
First we check that $u\in T$, $\{v_1,v_2,v_3\}$ is disjoint from $T$ and independent, and that $u$ is adjacent to $v_1$, $v_2$ and $v_3$. If the check fails, the guess is discarded.
We remove every neighbour of $u$ that is not in $\{v_1,v_2,v_3\}$. This gives us a graph $G'$. Note that $G'$ contains the $3$-part solution $F$ of $G$ if $u$ is the center of $F$ and $v_1,v_2,v_3$ are the center neighbours. Say $A$ is the center component of $F$. Then, by Lemma~\ref{l-sep2}, we have that $A-u$ consists of exactly three connected components $D_1$, $D_2$ and $D_3$, say $v_i\in V(D_i)$ for $i\in \{1,2,3\}$.

We first search for a maximum weight $3$-part solution of $G'$ with center~$u$ and center neighbours $v_1,v_2,v_3$ that is not full. This means that at least one of the connected components $D_1,D_2,D_3$ has size~$1$, so at least one of $v_1,v_2,v_3$ is adjacent only to $u$ in $F$. We consider each of the three options. We first examine the case where $v_1$ has degree~$1$ in $F$. In order to do this, we remove $u$ and $v_1$ from $G'$, and we also remove all neighbours of $v_1$ from $G'$. This gives us a graph $G_1$. It remains to solve {\sc Weighted Vertex Cut} for the instance that consists of $G_1$, with $v_2$ and $v_3$ as terminals and the restriction of $w$ to $V(G_1)$.
This takes polynomial time by Lemma~\ref{l-easy}. We repeat the same steps with respect to~$v_2$ and also with respect to $v_3$. From all the non-full $3$-part solutions found we remember one with the largest weight.

We now search for a maximum weight $3$-part solution of $G'$ with center~$u$ and center neighbours $v_1,v_2,v_3$ that is full. 
We remove every vertex not equal to $u$ that is adjacent to at least two vertices from $\{v_1,v_2,v_3\}$ from $G'$. We may do this, as $D_1,D_2,D_3$ are distinct connected components of $A-u$, so such a vertex will not be in any of $D_1,D_2$, or $D_3$. By Lemma~\ref{l-full}, every vertex in $F$ is adjacent to one of $v_1,v_2,v_3$.
Hence, we may also remove every vertex not adjacent to any vertex from $\{v_1,v_2,v_3\}$ from~$G'$. Finally, we remove every vertex of $T\setminus \{u\}$ from $G'$, as $u$ will be the only vertex of $F$ that belongs to $T$. We denote the resulting graph by $G''$; see also Fig.~\ref{f-gdo}. By construction, every vertex of $G''-\{u,v_1,v_2,v_3\}$ is adjacent to exactly one of $\{v_1,v_2,v_3\}$. Moreover, we constructed $G''$ in polynomial time.

\begin{figure}
\begin{center}
\begin{tikzpicture}[scale=1] \draw[color=black,fill=white] (-1,2) ellipse (0.6cm and 0.6cm); \draw[color=black,fill=white] (-1,0) ellipse (0.6cm and 0.6cm); \draw[color=black,fill=white] (-1,-2) ellipse (0.6cm and 0.6cm); \draw[dotted] (-3,2)--(-3,-2) to[out=60,in=300] (-3,2); \draw[dashed] (-3.25,2.7)--(-3.25,1.3)--(-0.15,1.3)--(-0.15,2.7)--(-3.25,2.7) (-3.25,-0.7)--(-3.25,0.7)--(-0.15,0.7)--(-0.15,-0.7)--(-3.25,-0.7) (-3.25,-2.7)--(-3.25,-1.3)--(-0.15,-1.3)--(-0.15,-2.7)--(-3.25,-2.7); \draw
(-5,0)--(-3,0)--(-1.2,0.566) (-3,0)--(-1.2,-0.566) (-5,0)--(-3,2)--(-1.2,2.566) (-3,2)--(-1.2,1.434) (-5,0)--(-3,-2)--(-1.2,-2.566) (-3,-2)--(-1.2,-1.434) (-1,1.3)--(-0.5,0.7) (-1.2,1.3)--(-1.2,0.7) (-1.2,-1.3)--(-0.5,-0.7) (-1,-1.3)--(-0.8,-0.7) (-0.15,-1.5) to[out=60,in=300] (-0.15,2) (-0.15,-2) to[out=60,in=300] (-0.15,2.5); \draw[fill=black] (-5,0) circle [radius=3pt] (-3,0) circle [radius=3pt] (-3,2) circle [radius=3pt] (-3,-2) circle [radius=3pt]; \node[above] at (-3,2.2) {$v_1$}; \node[above right] at (-3,0.2) {$v_2$}; \node[below] at (-3,-2.2) {$v_3$}; \node[below] at (-5.1,-0.25) {$u\in T$};
\end{tikzpicture}
\end{center}
\caption{The graph $G''$ constructed in the proof of Lemma~\ref{l-3part}. Note that $\{v_1,v_2,v_3\}$ is an independent set and that every vertex not in $\{u,v_1,v_2,v_3\}$ is adjacent to exactly one of $v_1,v_2,v_3$. The desired components $D_1$, $D_2$ and $D_3$ of $G''-u$ are, if they exist, subgraphs of the rectangular boxes. However, there might be edges between two vertices from two different boxes and we must deal with this situation.}\label{f-gdo}
\end{figure}
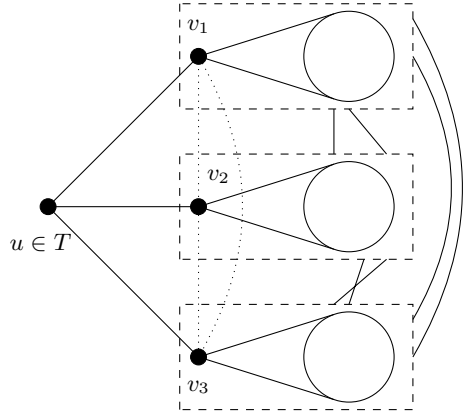

We will now consider all $O(n^3)$ options for choosing neighbours $x_1,x_2,x_3$ of $v_1,v_2,v_3$, respectively; note that these vertices are guessed to be in $D_1-v_1,D_2-v_2,D_3-v_3$, respectively.
As $F$ is full, these vertices $x_1,x_2,x_3$ exist by definition. For each chosen triple $x_1,x_2,x_3$ we do as follows. 
We first check if $\{x_1,x_2,x_3\}$ forms an independent set. This must hold as $D_1,D_2,D_3$ are distinct connected components of $A-u$. So, if the check fails, the guess is discarded.
From $G''$, we remove every neighbour of $v_1$ that is adjacent to at least one of $\{x_2,x_3\}$; every neighbour of $v_2$ that is adjacent to at least one of $\{x_1,x_3\}$; and every neighbour of $v_3$ that is adjacent to at least one of $\{x_1,x_2\}$. By Lemma~\ref{l-full}, we also remove for $i\in \{1,2,3\}$, every neighbour of $v_i$ that is not adjacent to $x_i$. We denote the resulting graph by $F^*$.

By construction, every vertex not in $\{u,v_1,v_2,v_3,x_1,x_2,x_3\}$ is not adjacent to $u$ and is adjacent to both vertices of exactly one of the pairs $\{v_1,x_1\}$, $\{v_2,x_2\}$ or $\{v_3,x_3\}$.
We claim that $F^*$ is the desired full $3$-part solution (for this particular branch).\footnote{For the remainder of the proof, the fact that $G$ is $(2P_1+P_4)$-free is again crucial. If $G$ is $5P_1$-free and thus $(3P_1+P_4)$-free, but not $(2P_1+P_4)$-free, the proof will not work. Namely, in that case, we would obtain the gadget of Papadopoulos and Tzimas~\cite{PT20}
for proving \NP-hardness of {\sc Subset Feedback Vertex Set} for $5P_1$-free graphs (in their proof the center~$u$ and the three center neighbours each have large weight while $T=\{u\}$ and a reduction from {\sc Vertex Cover} for $3$-partite graphs is used).} In order to see this, assume that $F^*-u$ is not the disjoint union of three complete graphs $D_1$, $D_2$, $D_3$. By construction, this means that one of the following two cases must hold:

\begin{figure}
\centering
\begin{minipage}[b]{0.4\textwidth}
\begin{tikzpicture}[scale=0.8] \draw (-2,-1)--(-2,0)--(0,1)--(0,-1) (0,1)--(2,0)--(2,-1) (-2,-2) to[out=60,in=300] (-2,0)
(0,-2) to[out=120,in=240] (0,0); \draw[very thick] (-2,-1)--(-2,-2)--(0,-2)--(0,-1); \draw[fill=white] (0,0) circle [radius=3pt] (-2,0) circle [radius=3pt] (2,0) circle [radius=3pt]; \draw[fill=black]
(0,1) circle [radius=3pt] (-2,-1) circle [radius=3pt] (0,-1) circle [radius=3pt] (2,-1) circle [radius=3pt] (-2,-2) circle [radius=3pt] (0,-2) circle [radius=3pt]; \node[above] at (0,1) {$u\in T$}; \node[left] at (-2,0) {$v_1$}; \node[right] at (0,0) {$v_2$}; \node[right] at (2,0) {$v_3$}; \node[left] at (-2,-1) {$x_1$}; \node[right] at (0,-1) {$x_2$}; \node[right] at (2,-1) {$x_3$}; \node[left] at (-2,-2) {$r$}; \node[right] at (0,-2) {$r'$};
\end{tikzpicture}
\end{minipage}
\qquad
\begin{minipage}[b]{0.2\textwidth}
\begin{tikzpicture}[scale=0.8] \draw (-2,-1)--(-2,0)--(0,1) (-3,-2.2)--(-3,0.2)--(-1,0.2)--(-1,-2.2)--(-3,-2.2) (2,0)--(2,-1) (-2.5,-2) to[out=120,in=220] (-2,0) (-1.5,-2) to[out=60,in=320] (-2,0); \draw[dotted] (-2.5,-2)--(-1.5,-2); \draw[very thick] (0,-1)--(0,1)--(2,0); \draw[dashed] (-2.5,-2)--(-2,-1)--(-1.5,-2); \draw[fill=white] (-2,0) circle [radius=3pt] (-2,-1) circle [radius=3pt] (2,-1) circle [radius=3pt]; \draw[fill=black] (2,0) circle [radius=3pt] (0,1) circle [radius=3pt] (0,0) circle [radius=3pt] (0,-1) circle [radius=3pt] (-2.5,-2) circle [radius=3pt] (-1.5,-2) circle [radius=3pt]; \node[above] at (0,1) {$u\in T$}; \node[left] at (-2,0) {$v_1$}; \node[right] at (0,0) {$v_2$}; \node[right] at (2,0) {$v_3$}; \node[left] at (-2,-1) {$x_1$}; \node[right] at (0,-1) {$x_2$}; \node[right] at (2,-1) {$x_3$}; \node[left] at (-2.5,-2) {$y$}; \node[right] at (-1.5,-2) {$y'$};\node[above] at (-2,0.3) {$$};
\end{tikzpicture}
\end{minipage}
\caption{The contradictions obtained in Lemma~\ref{l-3part}. On the left we have Case~(i) and on the right, Case (ii): in both cases we obtain an induced $2P_1+P_4$ (highlighted by the black vertices and thick edges).}\label{f-nn}
\end{figure}
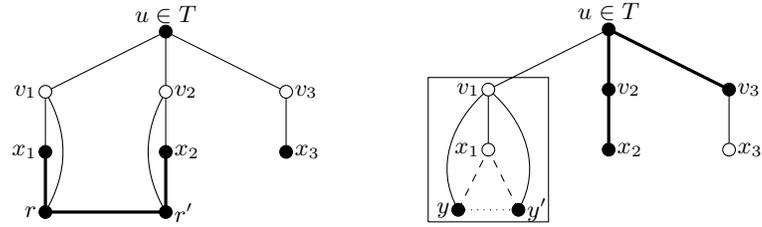

\begin{itemize}
\item [(i)] $F^*-u$ has two adjacent vertices $r$ and $r'$ that have a different neighbour in $\{v_1,v_2,v_3\}$, or
\item [(ii)] $F^*-u$ has a connected component that is not a complete graph.
\end{itemize}
For a contradiction, assume that Case~(i) holds. Say, $r$ is adjacent to $v_1$ and $r'$ is adjacent to $v_2$.
By construction of $F^*$, we find that $r$ is adjacent to $x_1$ but not to $x_2$ or $x_3$, whereas $r'$ is adjacent to $x_2$ but not to $x_1$ or $x_3$. 
Recall also that $\{x_1,x_2,x_3\}$ is an independent set and that $u$ is not adjacent to any vertex of $\{r,r',x_1,x_2,x_3\}$.
However, now the vertices $x_1,r,r',x_2,u,x_3$ induce a $2P_1+P_4$ in $F^*$, a contradiction as $F^*$ being an induced subgraph of $G$ is $(2P_1+P_4)$-free (see also Fig.~\ref{f-nn}).
Hence, Case (i) does not hold, and thus $F^*-u$ consists of three different connected components that contain the sets $\{v_1,x_1\}$, $\{v_2,x_2\}$ and $\{v_3,x_3\}$, respectively.

We conclude that Case~(ii) must hold if $F^*$ is not a full $3$-part solution.
We prove that this leads to another contradiction. Let $y,y'$ be two non-adjacent vertices in say the connected component of $F^*-u$ that contains $v_1$. As $v_1$ is adjacent to every vertex of that component, we find that $v_1$ is not equal to $y$ or $y'$. However, now $x_2,v_2,u,v_3,y,y'$ induce a $2P_1+P_4$, a contradiction; see also Fig.~\ref{f-nn}. 
Hence, we conclude that $F^*$ is indeed the desired full $3$-part solution. 

From all the full $3$-part solutions we found we remember one with the largest weight.
It remains to compare the maximum weight non-full solution with the maximum weight full solution and pick one with the largest weight.
The correctness of the algorithm follows from the arguments above. As the total number of branches is polynomial and processing each of them takes polynomial time as well, the running time of the algorithm is polynomial. \qed
\end{proof}

\subsection{Mim-Width}\label{s-mim}

We also need some known results that involve the mim-width of a graph. This width parameter was introduced by Vatshelle~\cite{Va12}. For the definition of mim-width we refer to~\cite{Va12}, as we do not need it here. A graph class~${\cal G}$ has {\it bounded} mim-width if there exists a constant $c$ such that every graph in ${\cal G}$ has mim-width at most~$c$. The mim-width of a graph class ${\cal G}$ is {\it quickly computable} if it is possible to compute in polynomial time a so-called branch decomposition for a graph~$G\in {\cal G}$ whose mim-width is bounded by some function in the mim-width of $G$. We can now state the aforementioned result of Bergougnoux, Papadopoulos and Telle in a more detailed way.

\begin{theorem}[\cite{BPT19}]\label{t-bpt19}
{\sc Weighted Subset Feedback Vertex Set} is polynomial-time solvable for every graph class whose mim-width is bounded and quickly computable.
\end{theorem}

\noindent
Belmonte and Vatshelle~\cite{BV13} proved that the mim-width of the class of permutation graphs is bounded and quickly computable. As $P_4$-free graphs form a subclass of the class of permutation graphs, we immediately obtain the following lemma.\footnote{It is well-known that $P_4$-free graphs have clique-width at most~$2$, and instead of Theorem~\ref{t-bpt19} we could have used a corresponding result for clique-width. We chose to formulate Theorem~\ref{t-bpt19} in terms of mim-width, as mim-width is a more powerful parameter than clique-width~\cite{Va12} and thus bounded for more graph classes.}

\begin{lemma}\label{l-p4}
The mim-width of the class of $P_4$-free graphs is bounded and quickly computable.
\end{lemma}

\noindent
For a graph class ${\cal G}$ and an integer $p\geq 0$, we let ${\cal G}+pv$ be the graph class that consists of all graphs that can be modified into a graph from ${\cal G}$ by deleting at most $p$ vertices. The following lemma follows in a straightforward way from a result of Vatshelle~\cite{Va12}.

\begin{lemma}\label{l-au}
If ${\cal G}$ is a graph class whose mim-width is bounded and quickly computable, then the same holds for the class ${\cal G}+pv$, for every constant $p\geq 0$.
\end{lemma}

\noindent
We will now explain the relevance of the above results in a slightly more general way (such that we can apply the next lemma in Section~\ref{s-2p1p4u} as well).
Let $G=(V,E)$ be an $(sP_1+P_4)$-free graph for some $s\geq 2$ and let $T\subseteq V$. Let $F$ be a $T$-forest of $G$. We define the {\it core} of $F$ as the set of vertices of $F$ that have at most $2s-1$ neighbours in $F$.
We say that $F$ is {\it core-complete} if the core of $F$ has no independent set of size at least~$s$; otherwise $F$ is
{\it core-incomplete}.\footnote{These notions are not meaningful if $s\in \{0,1\}$. Hence, we defined them for $s\geq 2$.}
We prove the following algorithmic result. 

\begin{lemma}\label{l-poss}
Let $s\geq 2$.
For an $(sP_1+P_4)$-free graph $G=(V,E)$ with a positive vertex weighting~$w$ and a set $T\subseteq V$, it is possible to find a core-incomplete solution of maximum weight in polynomial time.
\end{lemma}

\begin{proof}
Let $|V|=n$.
A core-incomplete solution $F$ has a core that contains an independent set $U=\{u_1,\ldots,u_s\}$ of size~$s$. By definition, core vertices have degree at most~$2s-1$ in $F$. We consider all $O(n^s)$ options of choosing the vertices $u_1,\ldots,u_s$. 
We check if $U$ is an independent set. If the check fails, the guess is discarded.
For each choice, we consider all $O(n^{s(2s-1)})$ options of choosing a set~$Z$ of size at most $s(2s-1)$ that contains the (at most $2s-1$) neighbours of each $u_i$
$(1\leq i \leq s)$ in $F$. For each choice of  $Z$ we now do as follows. We delete every neighbour of every $u_i$ that is not in $Z$ from $G$, that is, we delete the set $N(U)\setminus Z$. Let $G'$ be the new graph; see also Fig.~\ref{f-poss}. Note that $N_{G'}(U)=Z$.
Let $w'$ be the restriction of $w$ to $V(G')$ and let $T'=T\cap V(G')$.

As $G$ is $(sP_1+P_4)$-free, $G'$ is also $(sP_1+P_4)$-free.
As $U$ is an independent set of size~$s$ and $N_{G'}(U)=Z$, this means that $G'-(U\cup Z)$ is $P_4$-free. Let ${\cal G}$ be the class of $P_4$-free graphs. Then, as $|U\cup Z|\leq s(2s-1)+s=2s^2$, we have that $G'$ belong to ${\cal G}+2s^2v$.
As the mim-width of the class ${\cal G}$ of $P_4$-free graphs is bounded and quickly computable by Lemma~\ref{l-p4}, we can apply Lemma~\ref{l-au} and then we can use Theorem~\ref{t-bpt19} on $(G',w',T')$. Since the number of branches is $O(n^{2s^2})$, the running time of our algorithm is polynomial.
\qed
\end{proof}

\begin{figure}
\begin{center}
\begin{tikzpicture}[scale=1] \draw[color=black,fill=white] (1,0) ellipse (1cm and 2.5cm); \draw (-2.3,2.3)--(-2.6,2.3)--(-2.6,-2.3)--(-2.3,-2.3) (2.3,2.8)--(2.6,2.8)--(2.6,-2.8)--(2.3,-2.8) (3,1.7)--(3,-1.7)--(5,-1.7)--(5,1.7)--(3,1.7) (-2,2)--(0.4,2)
(-2,1.5)--(0.2,1.5) (-2,1)--(0.08,1) (-2,0.5)--(0.028,0.5) (-2,0)--(0,0) (-2,-2)--(0.4,-2) (-2,-1.5)--(0.2,-1.5) (-2,-1)--(0.08,-1) (-2,-0.5)--(0.028,-0.5); \draw[fill=black] (-2,2) circle [radius=2pt] (-2,1.5) circle [radius=2pt] (-2,-2) circle [radius=2pt] (-2,1) circle [radius=1.5pt] (-2,0.5) circle [radius=1pt] (-2,0) circle [radius=1pt] (-2,-0.5) circle [radius=1pt] (-2,-1) circle [radius=1pt](-2,-1.5) circle [radius=1pt]; \node[left] at (-2.7,0) {$s$}; \node[above] at (-2, 2.4) {$U$}; \node[above] at (1,2.6) {$Z$}; \node[right] at (2.7,-2.8) {$\leq s(2s-1)$}; \node[above] at (4,1.8) {$\;\;\;\;\;V(G')\setminus (U\cup Z)$}; 
\end{tikzpicture}
\end{center}
\caption{The graph $G'$ from the proof of Lemma~\ref{l-poss}.}\label{f-poss}
\end{figure}
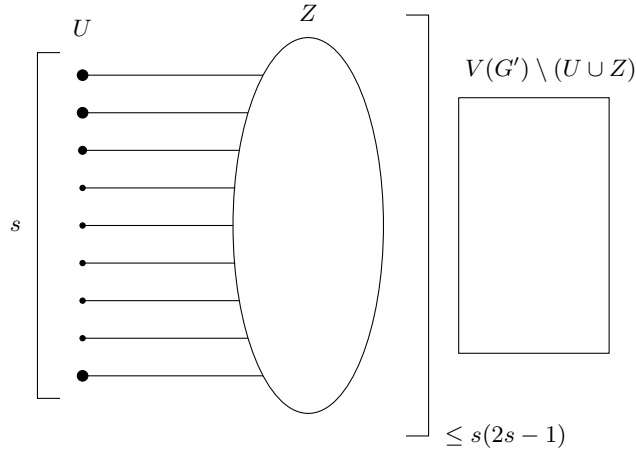

\subsection{The Algorithm}\label{s-algo}

In this section we present our algorithm for {\sc Weighted Subset Feedback Vertex Set} restricted to $(2P_1+P_4)$-free graphs.  We first need to prove one more structural lemma for core-complete solutions. We prove this lemma for any value $s\geq 2$, such that we can use this lemma in the next section as well. However, for $s=2$ we have a more accurate upper bound on the size of the core.

\begin{lemma}\label{l-two}
For some $s\geq 2$, let $G=(V,E)$ be an $(sP_1+P_4)$-free graph. Let $T\subseteq V$. Let $F$ be a core-complete $T$-forest of $G$ such that $T\cap V(F)\neq \emptyset$. Then the core of $F$ 
contains every vertex of $T\cap V(F)$, and $T\cap V(F)$ has size at most $2s-2$. If $s=2$, the core of $F$ is a clique of size at most~$2$ (in this case $T\cap V(F)$ has size at most~$2$ as well).
\end{lemma}

\begin{proof}
Consider a vertex $u\in T\cap V(F)$.
For a contradiction, assume that $u$ does not belong to the core of $F$. Then $u$ has at least $2s$ neighbours in $F$. Let $V_u=\{v_1,\ldots,v_p\}$ for some $p\geq 2s$ be the set of neighbours of $u$ in $F$.

Let $A$ be the connected component of $F$ that contains $u$. As $F$ is a $T$-forest, $A-u$ consists of $p$ connected components $D_1,\ldots, D_p$ such that $v_i\in V(D_i)$ for $i\in \{1,\ldots,p\}$. In particular, this implies that $V_u=\{v_1,\ldots,v_p\}$ must be an independent set. As the core of $F$ has no independent set of size~$s$, this means that at most $s-1$ vertices of $V_u$ may belong to the core of $F$. Recall that $p\geq 2s$.
Hence, 
we may assume without loss of generality that $v_1,\ldots,v_{s+1}$ do {\it not} belong to the core of $F$. This means that $v_1,\ldots, v_{s+1}$ each have degree at least~$2s$ in $A$. Hence, for $i\in \{1,\ldots,s+1\}$, vertex $v_i$ is adjacent to some vertex $w_i$ in $D_i$. As $s\geq 2$, we have that $2s>s+1$ and hence, vertex $v_{s+2}$ exists.
However, now the vertices $w_1,v_1,u,v_{s+2},w_2,w_3,\ldots,w_{s+1}$ induce an $sP_1+P_4$, a contradiction (see also Fig.~\ref{f-two}).

From the above, we conclude that every vertex of $T\cap V(F)$ belongs to the core of $F$.
As $F$ is a $T$-forest, $T\cap V(F)$ induces a forest, and thus a bipartite graph. As $F$ is core-complete, every independent set in the core has size at most $s-1$. Hence, $T\cap V(F)$ has size at most $2(s-1)=2s-2$.

Now suppose that $s=2$. As $F$ is core-complete, the core of $F$ must be a clique. As the core of $F$ contains $T\cap V(F)$ and $T\cap V(F)$ induces a forest, this means that the core of $F$, and thus also $T\cap V(F)$, has  size at most~$2$. This completes the proof of the lemma.
\qed
\end{proof}

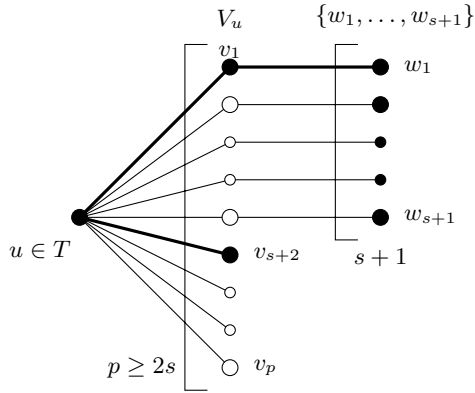
\begin{figure}
\begin{center}
\begin{tikzpicture}[scale=1] \draw[very thick] (2,2)--(0,2)--(-2,0)--(0,-0.5); \draw(-2,0)--(0,1.5)--(2,1.5) (-2,0)--(0,1)--(2,1) (-2,0)--(0,0.5)--(2,0.5) (-2,0)--(0,0)--(2,0) (-0.3,2.3)--(-0.6,2.3)--(-0.6,-2.3)--(-0.3,-2.3) (1.7,2.3)--(1.4,2.3)--(1.4,-0.3)-- (1.7,-0.3) (-2,0)--(0,-1) (-2,0)--(0,-1.5) (-2,0)--(0,-2); \draw[fill=black] (-2,0) circle [radius=3pt] (0,2) circle [radius=3pt] (0,-0.5) circle [radius=3pt] (2,2) circle [radius=3pt] (2,1.5) circle [radius=3pt] (2,1) circle [radius=2pt] (2,0.5) circle [radius=2pt] (2,0) circle [radius=3pt]; \draw[fill=white] (0,1.5) circle [radius=3pt] (0,1) circle [radius=2pt] (0,0.5) circle [radius=2pt] (0,0) circle [radius=3pt] (0,-1) circle [radius=2pt] (0,-1.5) circle [radius=2pt] (0,-2) circle [radius=3pt]; \node[below left] at (-2,-0.2) {$u\in T$}; \node[left] at (-0.6,-2) {$p\geq 2s$}; \node[below] at (2,-0.3) {$s+1$}; \node[above] at (0,2.4) {$V_u$}; \node[above] at (2.2,2.4) {$\{w_1,\ldots,w_{s+1}\}$}; \node[right] at (2.2,2) {$w_1$}; \node[right] at (2.2,0) {$w_{s+1}$}; \node[above] at (0,2) {$v_1$};  \node[right] at (0.2,-0.5) {$v_{s+2}$}; \node[right] at (0.2,-2) {$v_p$}; \end{tikzpicture}
\end{center}
\caption{An example of the contradiction obtained in Lemma~\ref{l-two}: the assumption that a vertex $u\in T$ does not belong to the core of a core-complete solution leads to the presence of an induced $sP_1+P_4$ (highlighted by the black vertices and thick edges).}\label{f-two}
\end{figure}

\noindent
By using the above results and the results from Sections~\ref{s-special} and~\ref{s-mim}, we are now able to prove our main result. 

\begin{theorem}\label{t-mm}
{\sc Weighted Subset Feedback Vertex Set} is polynomial-time solvable for $(2P_1+P_4)$-free graphs.
\end{theorem}

\begin{proof}
Let $G=(V,E)$ be a $(2P_1+P_4)$-free graph, and let $T$ be some subset of~$V$. Let $w$ be a positive vertex weighting of $G$. We aim to find a maximum weight $T$-forest~$F$ for $(G,T,w)$ (recall that we call $T$-forests solutions for our problem).
As $s=2$, the core of $F$ is, by definition, the set of vertices of $F$ that have maximum degree at most~$3$ in $F$.
 
We first compute a core-incomplete solution of maximum weight; this takes polynomial time by Lemma~\ref{l-poss} (in which we set $s=2$). We will now compute in polynomial time a core-complete solution $F$ of maximum weight for $(G,T,w)$. We then compare the weights of the two solutions found to each other and pick one with the largest weight.

By Lemma~\ref{l-two}, it holds for every core-complete solution $F$ that $T\cap V(F)$ belongs to the core of $F$, and moreover that $|T\cap V(F)|\leq 2$. 
We first compute a core-complete solution $F$ with $|T\cap V(F)|\leq 1$ of maximum weight.
As  $T\cap V(F)$ belongs to the core of $F$, we find that if $|T\cap V(F)|=1$, say $T\cap V(F)=\{u\}$ for some $u\in T$, then $u$ has maximum degree at most~$3$ in $F$. Hence, in the case where  $|T\cap V(F)|\leq 1$, it suffices to compute a $\leq$$1$-part solution, $2$-part solution and $3$-part solution for $(G,T,w)$ of maximum weight and to remember one with the largest weight. By Lemmas~\ref{l-1part},~\ref{l-2part} and~\ref{l-3part}, respectively, this takes polynomial time.

It remains to compute a core-complete solution $F$ with $|T\cap V(F)|=2$ of maximum weight. 
By Lemma~\ref{l-two}, it holds for every such solution $F$ that both vertices of $T\cap V(F)$ are adjacent and are the only vertices that belong to the core of $F$.

We consider all $O(n^2)$ possibilities of choosing two adjacent vertices of $T$ to be the two core vertices of $F$. Consider such a choice of adjacent vertices $u_1,u_2$. So, $u_1$ and $u_2$ are the only vertices of degree at most~$3$ in the solution $F$ that we are looking for and moreover, all other vertices of $T$ do not belong to $F$.

Suppose one of the vertices $u_1,u_2$ has degree~$1$ in $F$. First let this vertex be $u_1$. Then we remove $u_1$ and all its neighbours except for $u_2$ from $G$. Let $G'$ be the resulting graph. Let $T'= T \setminus (\{u_1\}\cup (N(u_1) \setminus \{u_2\}))$, and let $w'$ be the restriction of $w$ to $G'$. We now compute for $(G',w',T')$, a $\leq$$1$-part solution and  $2$-part solution of maximum weight with $u_2$ as center.
By Lemmas~\ref{l-1part} and~\ref{l-2part}, respectively, this takes polynomial time.\footnote{Strictly speaking, this statement follows from the proofs of these two lemmas, as we have fixed $u_2$ as the center.} We then add $u_1$ back to the solution to get a solution for $(G,w,T)$. We do the same steps with respect to $u_2$. In the end we take a solution with largest weight.

So from now on, assume that both $u_1$ and $u_2$ have degree at least~$2$ in $F$. We first argue that in this case both $u_1$ and $u_2$ have degree exactly~$2$ in $F$. For a contradiction, suppose $u_1$ has degree~$3$ in $F$ (recall that $u_1$ has degree at most~$3$ in $F$). 
Let $v_1$ and $v_1'$ be two distinct neighbours of $u_1$ in $V(F)\setminus \{u_2\}$. Let $v_2$ be a neighbour of $u_2$ in $V(F)\setminus \{u_1\}$.
As $F$ is a $T$-forest, $v_1,v_1',v_2$ belong to distinct connected components $D_1$, $D_1'$ and $D_2$, respectively, of $F-\{u_1,u_2\}$. As the core of $F$ consists of $u_1$ and $u_2$ only, $v_1,v_1',v_2$ each have a neighbour $x_1,x_1',x_2$ in $D_1$, $D_1'$ and $D_2$, respectively.
However, now $x_2,v_2,u_2,u_1,x_1,x_1'$ induce a $2P_1+P_4$ in $F$ and thus also in $G$, a contradiction; see also Fig.~\ref{f-mm}.

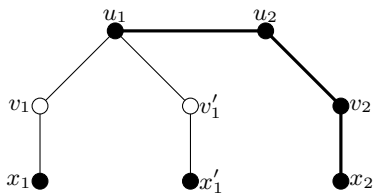
\begin{figure}
\begin{center}
\begin{tikzpicture}[scale=1] \draw (-2,-1)--(-2,0)--(-1,1)--(0,0)--(0,-1); \draw[very thick] (-1,1)--(1,1)--(2,0)--(2,-1); \draw[fill=white] (0,0) circle [radius=3pt] (-2,0) circle [radius=3pt]; \draw[fill=black] (-2,-1) circle [radius=3pt] (2,0) circle [radius=3pt] (0,-1) circle [radius=3pt] (2,-1) circle [radius=3pt] (-1,1) circle [radius=3pt] (1,1) circle [radius=3pt]; \node[above] at (-1,1) {$u_1$}; \node[above] at (1,1) {$u_2$}; \node[left] at (-2,0) {$v_1$}; \node[right] at (0,0) {$v'_1$}; \node[right] at (2,0) {$v_2$}; \node[left] at (-2,-1) {$x_1$}; \node[right] at (0,-1) {$x'_1$}; \node[right] at (2,-1) {$x_2$};  \end{tikzpicture}
\end{center}
\caption{The situation in Theorem~\ref{t-mm} where $u_1$ has degree at least~$3$ in $F$ and $u_2$ has degree~$2$ in $F$; this leads to the presence of an induced $2P_1+P_4$ (highlighted by the black vertices and thick edges).}\label{f-mm}
\end{figure}

From the above we conclude that each of $u_1$ and $u_2$ has exactly one other neighbour in $F$. Call these vertices $v_1$ and $v_2$, respectively. We consider all $O(n^2)$ possibilities of choosing $v_1$ and $v_2$. As $F$ is a $T$-forest, $G-\{u_1,u_2\}$ consists of two connected components $D_1$ and $D_2$, such that $v_1$ belongs to $D_1$ and $v_2$ belongs to $D_2$.

Let $G'$ be the graph obtained from $G$ by removing every vertex of $T$, every neighbour of $u_1$ except $v_1$ and every neighbour of $u_2$ except $v_2$. Let $w'$ be the restriction of $w$ to $G'$.
Then, it remains to solve {\sc Weighted Vertex Cut} for the instance $(G',v_1,v_2,w')$. By Lemma~\ref{l-easy}, this can be done in polynomial time. Out of all the solutions found for different pairs $u_1,u_2$ we take one with the largest weight. Note that we found this solution in polynomial time, as the number of branches is $O(n^4)$. 

As mentioned we take a solution of maximum weight from all the solutions found in the above steps. The correctness of our algorithm follows from the fact that we exhaustively considered all possible situations. Moreover, the number of situations is polynomial and processing each situation takes polynomial time. Hence, the running time of our algorithm is polynomial.
\qed
\end{proof}

\section{The Unweighted Variant}\label{s-2p1p4u}

In this section, we present our polynomial-time algorithm for {\sc Subset Feedback Vertex Set} on $(sP_1+P_4)$-free graphs for every $s\geq 0$. As this problem is a special case of {\sc Weighted Subset Feedback Vertex Set} (namely when $w\equiv 1$), we can use some of the structural results from the previous section.

\begin{theorem}\label{t-unw}
{\sc Subset Feedback Vertex Set} is polynomial-time solvable on $(sP_1+P_4)$-free graphs for every $s\geq 0$.
\end{theorem}

\begin{proof}
Let $G=(V,E)$ be an $(sP_1+P_4)$-free graph for some integer $s$, and let $T\subseteq V$. Let $|V|=n$. As the class of $(sP_1+P_4)$-free graphs is a subclass of the class of $((s+1)P_1+P_4)$-free graphs, we may impose any lower bound on $s$; we set $s\geq 2$. 
We aim to find a $T$-forest $F$ of $G$ of maximum size (recall that we call $T$-forests solutions for our problem).

We first compute a maximum-size core-incomplete solution for $(G,T)$. By Lemma~\ref{l-poss}, this takes polynomial time. It remains to compare the size of this solution with a maximum-size core-complete solution, which we compute below.

By Lemma~\ref{l-two}, we find that $T\cap V(F)$ has size at most $2s-2$ for every core-complete solution~$F$; see also Fig.~\ref{f-scheme}. We consider all $O(n^{2s-2})$ possibilities of choosing the vertices of $T\cap V(F)$. For each choice of $T\cap V(F)$ we do as follows. We note that the set of vertices of $G-T$ that do not belong to $F$ has size at most $|T\cap V(F)|$; otherwise $F'=V\setminus T$ would be a larger solution than~$F$. Hence, we can consider all $O(n^{|T\cap V(F)|})=O(n^{2s-2})$ possibilities of choosing the set of vertices of $G-T$ that do not belong to $F$, or equivalently, of choosing the set of vertices of $G-T$ that {\it do} belong to $F$. In other words, we guessed $F$ by brute force, and the number of guesses is $O(n^{4s-4})$. In the end we found in polynomial time a maximum-size core-complete solution. We compare it with the maximum-size core-incomplete solution found above and pick one with the largest size. 
\qed
\end{proof}

\begin{figure}
\begin{center}
\begin{tikzpicture}[scale=0.4] \draw (-10,5)--(13,5)--(13,-5)--(-10,-5)--(-10,5) (2,5)--(2,-5) (-10,-1)--(13,-1); \node[above] at (-4.5,5) {$V(F)=V\setminus S$}; \node[above] at (5.5,5) {$S$}; \node[left] at (-10,2) {$T$}; \node[left] at (-10,-3) {$V\setminus T$}; \node at (-4.5,3) {$T\cap V(F)$}; \node at (-4.5,1) {size: $\leq 2s-2$}; \node at (5.5,2) {$T\cap S$}; \node at (-4.5,-3) {$V(F)\setminus T$}; \node at (5.5,-2) {$S\setminus T$}; \node at (5.5,-4) {$\;\;\;\;\;\;\;\;\;\;\;\;\;\;$ size: $\leq |T\cap V(F)|\leq 2s-2$}; \end{tikzpicture}
\end{center}
\caption{The situation in Theorem~\ref{t-unw} where we search for a largest core-complete solution for $(G,T)$.}\label{f-scheme}
\end{figure}

\section{Conclusions}\label{s-con}

By combining known hardness results with new polynomial-time results, we completely classified the complexities of {\sc Weighted Subset Feedback Vertex Set} and {\sc Subset Feedback Vertex Set} for $H$-free graphs. We recall that the classical versions {\sc Weighted Feedback Vertex Set} and {\sc Feedback Vertex Set} are not yet completely classified due to missing cases when $H$ is a linear forest (see Theorem~\ref{t-kn}). 

 We finish our paper with a brief discussion on the variant where instead of intersecting every $T$-cycle, a solution only needs to intersect every $T$-cycle of {\it odd} length. These two problems are called  {\sc Weighted Subset Odd Cycle Transversal} and {\sc Subset Odd Cycle Transversal}, respectively. 
 
So far, both these problems behave in exactly the same way on $H$-free graphs as their feedback vertex set counterparts (see~\cite{BJPP20} and~\cite{BJP21}). So, the only open cases for {\sc Weighted Subset Odd Cycle Transversal} on $H$-free graphs are the ones where $H\in \{2P_1+P_3,P_1+P_4,2P_1+P_4\}$ and the only open cases for {\sc Subset Odd Cycle Transversal} on $H$-free graphs are the ones where $H=sP_1+P_4$ for some $s\geq 1$. Since solutions $F$ for these problems may only contain vertices of $T$ of high degree, we can no longer use the proof technique from this paper, and new ideas are needed. 
 
We note, however, that complexity dichotomies of {\sc Weighted Subset Odd Cycle Transversal} and {\sc Subset Odd Cycle Transversal} do not have to coincide with those in Theorems~\ref{t-dicho} and~\ref{t-dicho2} for their feedback vertex set counterparts. After all, the complexities of the corresponding classical versions may not coincide either. Namely, it is known that {\sc Odd Cycle Transversal} is \NP-complete for $(P_2+P_5,P_6)$-free graphs~\cite{DFJPPP19}, and thus for $(P_2+P_5)$-free graphs and $P_6$-free graphs, whereas for {\sc Feedback Vertex Set} such a hardness result is unlikely: for every linear forest $H$, {\sc Feedback Vertex Set} is quasipolynomial-time solvable on $H$-free graphs~\cite{GLPPR21}.

\end{document}